\makeatletter
\newcommand*{\rom}[1]{\expandafter\@slowromancap\romannumeral #1@}
\makeatother
\documentclass[journal,12pt,onecolumn,draftcls]{IEEEtran}
\IEEEoverridecommandlockouts
\ifCLASSINFOpdf
\else
\fi
\usepackage{lipsum}
\usepackage{graphicx}
\usepackage[utf8]{inputenc}
\usepackage{lmodern,textcomp}
\usepackage{graphics}
\usepackage{multicol}
\usepackage{lipsum}
\usepackage{color}
\usepackage[cmex10]{amsmath}
\usepackage{amsthm}
\usepackage{amssymb}
\usepackage{mathrsfs}
\usepackage{mathtools}
\usepackage{amsbsy}
\usepackage{mathrsfs}
\usepackage{setspace}
\usepackage{float}
\usepackage{graphicx}
\usepackage{psfrag}
\usepackage{lettrine}
\usepackage{newclude}
\usepackage[normalem]{ulem}
\usepackage{latexsym}
\usepackage{algpseudocode}
\usepackage{algorithm,algpseudocode}
\usepackage{algorithmicx}
\usepackage{graphicx,amsthm,amsmath,amssymb}
\usepackage{multirow}
\usepackage{caption}
\usepackage{rotating}
\usepackage{booktabs}
\usepackage{amsmath}
\usepackage{wasysym}
\usepackage{mathrsfs}
\usepackage{bbm} 
\usepackage{filecontents}
\usepackage{amsmath,epsfig,cite,amsfonts,amssymb,psfrag,color}
\usepackage{apptools}
\usepackage{graphicx}
\usepackage[table]{xcolor}
\usepackage{chngcntr}
\usepackage{blindtext}

\ifCLASSOPTIONcompsoc
\usepackage[caption=false,font=normalsize,labelfon
t=sf,textfont=sf]{subfig}
\else
\usepackage[caption=false,font=footnotesize]{subfig}

\AtAppendix{\counterwithin{lemma}{section}}

\usepackage{accents}
\newlength{\dhatheight}

\allowdisplaybreaks

\newtheorem{theorem}{Theorem}

\newtheorem{proposition}{Proposition}
\newtheorem{corollary}{Corollary}

\graphicspath{{figures/}}
\allowdisplaybreaks

\usepackage[skins,theorems]{tcolorbox}
\tcbset{highlight math style={shrink tight,extrude by=0.55mm,colframe=red,
  boxrule=0.4pt,frame style={opacity=0.6}}}
  
\begin{document}
\title{Average Peak Age of Information Analysis for Wireless Powered Cooperative Networks} 

\author{\IEEEauthorblockN{\normalsize
    Yasaman Khorsandmanesh$^\dagger$, Mohammad Javad Emadi$^\dagger$, and Ioannis Krikidis$^\ddagger$}
  \IEEEauthorblockA{\\ $^\dagger$Electrical Engineering Department, Amirkabir University of Technology, Tehran, Iran\\ $^\ddagger$Computer and Electrical Engineering  Department, University of Cyprus, Nicosia, Cyprus
  \\  E-mails:  $^\dagger$\{y.khorsandmanesh, mj.emadi\}@aut.ac.ir, $^\ddagger$krikidis@ucy.ac.cy}}
\maketitle
\begin{abstract}
Age of information (AoI), a metric to analyse data freshness, is of interest for time-critical applications in upcoming wireless networks. Besides, wireless power transfer (WPT) is also essential to prolong lifetime of a wireless sensor network. Thus, we study a flat-fading wireless powered cooperative network, wherein a source and a relay charge their finite-sized capacitors by harvesting energy from a remote power station. For two conventional decode-and-forward (DF) and amplify-and-forward (AF) relaying, average AoI and peak AoI (PAoI) are studied to show how the randomness of the data/power transfer channel and size of the capacitors affect the age metrics, and when utilizing a cooperative transmission is more beneficial than direct one. It is shown that, although the power-limited relay imposes more delay overhead, the age metrics are improved in some circumstances. Even for the DF scheme with more waiting time for signal processing, we may achieve lower average PAoI due to enhancing the end-to-end communications reliability. Furthermore, special cases of the proposed WPT cooperative system is discussed, e.g. cases of full-power nodes and one-shot successful data transmission. Finally, numerical results are provided to highlight effects of transmission power, size of the capacitors, and path loss exponent on the age metrics.
\end{abstract}
\IEEEpeerreviewmaketitle
\newpage
\section{Introduction}\label{sec1introduc}
\IEEEPARstart{T}{\lowercase{he}} global evolutionary wireless networks are connecting everything to support a wide range of use cases with various requirements and device limitations. Among those use cases, machine-type communications to support variations of Internet of Things (IoT) networks have gained a lot of interest from research and industrial sides. For instance, four IoT connectivity segments are defined in the 5G era by Ericsson \cite{ ericsson2020IoT}; massive IoT, critical IoT, industrial automation IoT, and broadband IoT. For most of the cases, one needs low data rates, low-cost devices with long-lifetime battery, and more importantly low-latency communications to support time-sensitiveness applications. Thus, analysing time-critical/sensitive IoT and wireless sensor networks (WSNs) with low-cost and battery-limited devices has been in the spotlight of studies in the next emerging applications beyond 5G and 6G \cite{leung20195g,qi2020integration}. One of the promising techniques to charge capacitors of the sensors is wireless power transfer (WPT), which is first introduced in \cite{varshney2008transporting}, to power the next generation of wireless networks. Various variations of WPT techniques over the radio frequency (RF) signals and harvesting energy are studied in the literature \cite{lu2014wireless}. On the other hand, to evaluate the performance of the time-critical applications, another metric to measure the freshness of data is required. 

Although in communication networks, the conventional goals are to maximize the achievable throughput and energy efficiency of the networks, for the time-critical applications measuring \emph{freshness} of data is important as well. Thus, the age of information (AoI) metric, which nicely measures freshness of data, is introduced in \cite{yates}. The AoI metric is significantly different from the classical delay and latency metrics, and captures the freshness of data over time at the destination node. For a given WSN, age is defined as the time elapsed since the most recent status update was generated at the sensor and reached to the destination node successfully. Different variations of the AoI metric is introduced and analysed in the literature. For instance, age metrics are evaluated for a conventional communication system models in \cite{kam2013age,yates2012real,he2016optimizing}, for multi-hop networks in \cite{talak2017minimizing,maatouk2018age,arafa2019timely}, and for cognitive radio networks in \cite{valehi2017maximizing}. Specifically, \cite{li2020age} considers a three-node cooperative status update system and propose an age-oriented opportunistic relaying (AoR) protocol to reduce the AoI of the system. One of the main variations of age metric is the peak age of information (PAoI), which provides information about the maximum value of AoI for each status update \cite{costa2014age}. As the PAoI captures the extent to which the update information is no longer fresh and is easier to trace and analytically evaluate, it has been considered as an efficient metric to investigate the freshness of the delivered information in a cooperative network \cite{costa2016age}.
Besides, the results show that performance characteristics of the average AoI and average PAoI are almost the same, and the difference between the two metrics is small.

Due to widespread applications of IoT and WSN with low-cost and limited-capacity battery devices, studying freshness of data in such networks has gained a lot of research interest. Therefore, analysing the age metrics for such WPT communication systems, wherein the nodes harvest energy from the ambient wireless power sources, has become a trended research area with lots of applications. For instance, optimal transmission policies subject to battery capacity constraints are studied in \cite{ref1,ref2,arafa2019age,bacinoglu2018achieving,bacinoglu2019optimal,ozel2020timely}. An energy-harvesting monitoring node equipped with a finite-size battery and collecting status updates from multiple heterogeneous information sources is presented, and AoI is analysed in  \cite{diversity}. Moreover, policies that minimize the average AoI and the role of information source diversity on system performance are studied.  Authors in \cite{leng2019age} have studied AoI minimization in a cognitive radio network wherein an energy harvesting secondary user opportunistically transmits status updates to its destination.  In  \cite{kirikidis}, the authors consider a single source--destination link, in which the source has a capacitor to charge its battery; the average AoI is derived, and the optimal value of the capacitor is computed. In \cite{uplink}, a two-way data exchanging model is presented, wherein the master node also transfers energy to enable the uplink transmission as well. Although analysing the freshness of data in WPT networks seems crucial, other related works cannot be found in the literature. 

Since utilizing cooperative relays in wireless communication networks provides higher levels of link reliability \cite{el2011network}, use of relaying schemes has been extensively studied in the literature from different perspectives \cite{levin2012amplify,relay1,relay2,relay3,relay4,chen2015harvest,liu2014two,zeng2015full,chen2014wireless}. In \cite{levin2012amplify} performance of multi-hop relay channels under two conventional amplify-and-forward (AF) and decode-and-forward (DF) schemes are compared. Other works have studied effects of relaying on the performance of a wireless network  \cite{relay1,relay2}. Besides, resource allocations and relay selections are studied in \cite{relay3,relay4}. Moreover, to prolong the lifetime and improve the energy efficiency of IoT and WSN networks,  source and relay can both harvest energy from RF transmitted signals from a base station, e.g. a harvest-then-cooperate protocol for wireless powered relay networks is investigated in \cite{chen2015harvest}. In \cite{liu2014two}, a two-way AF relaying WPT network is investigated, in which an energy-constrained relay  harvests energy from the ambient. Also, a full-duplex relay WPT system and WPT-based two-way relay networks are studied in \cite{zeng2015full} and \cite{chen2014wireless}, respectively.

To the best of our knowledge, considering the average PAoI metric in a \textit{cooperative wireless powered} sensor network has \emph{not} been addressed previously in the literature. In such systems, age evolution depends on two main factors; \emph{energy arrival rate} and \emph{successful transmission probability}. Both factors highly depend on channel statistics and its variations. At first glance, one may expect that utilizing a cooperative relay causes a higher level of reliability at the cost of higher average age metrics. To investigate this issue, we assume a two-hop relay WPT-based system, wherein the source and relay are equipped with a finite-size capacitor and capture their required transmission powers from the received RF signal transmitted from a power station to charge them remotely. Due to the randomness of the channel, the charging time is modeled by a random variable. Also, data transmission over each communications link is successful when the signal-to-noise ratio becomes higher than a threshold; otherwise, the transmitter must re-transmit the signal. This random re-transmission time is also investigated in the paper. For DF- and AF-based schemes, average AoI and PAoI are investigated, and various special cases are also discussed. Our results indicate that using a cooperative relay can dramatically improve the average PAoI in some situations, and sometimes it is better to send the packets directly to the destination, i.e. without using a relay. Finally, numerical results are presented to highlight the performance of the system under various circumstances.  


The organization of the paper is as follows. Section \ref{sec2sysmodel} presents an overview of the proposed channel model, wireless power transfer, data transmission, and AoI metrics. Average AoI and average PAoI are analysed for the DF-- and AF--based cooperative systems in Section \ref{sec3AOI}. In Section \ref{sec4S}, special cases of a proposed system are thoroughly discussed. 
Section \ref{sec5Ssim} presents numerical results to evaluate the system performance for various schemes. Finally, We conclude the paper in Section \ref{sec6con}.

\textit{Notations}: Random variables (R.Vs) are presented by capital letters and their realisations are denoted by small letter cases. Probability over a set $\mathcal{A}$ is denoted by $\mathbb{P}_r\{\mathcal{A}\}$,  $\mathbb{E}[X]$ represents the statistical average of R.V $X$, and $X\sim \mathcal{CN}(0,\sigma^2)$ denotes a circularly symmetric zero-mean Gaussian R.V with variance $\sigma^2$. The ordinary and incomplete Gamma functions are denoted by $\Gamma(.)$ and $\Gamma(.,.)$, respectively.

\begin{figure}[!t]
        \centering
         \psfrag{gr}{\hspace*{46mm} \ $\sqrt{d_{r,p\textcolor{white}{d}
         }^{-\alpha}}g_{r,p}$}
         \psfrag{gs} {\hspace*{-6mm}\ \ $\sqrt{d_{s,p\textcolor{white}{d}
         }^{-\alpha}}g_{s,p}$}
         \psfrag{br} {\hspace*{-5mm}  $B_r$}
         \psfrag{bs} {\hspace*{-5mm}   \ $B_s$}
        \psfrag{hr,s} { \hspace*{-9mm} $\sqrt{B_sd_{r,s\textcolor{white}{d}}^{-\alpha}} h_{r,s}$}
        \psfrag{hd,r} { \hspace*{35mm}\ $\small{\sqrt{B_rd_{d,r\textcolor{white}{r}}^{-\alpha}}h_{d,r}}$}
        \psfrag{relay} {\hspace*{-5mm} relay}
        \psfrag{energy flow} {\hspace*{-9mm}\ Energy flow}
        \psfrag{information flow} {\hspace*{-9mm}\ Information flow}
        \psfrag{Power Station} {\hspace*{-8mm}\ power station}
        \psfrag{source} {\hspace*{-5mm} source}
        \psfrag{desti} {\hspace*{-10mm} \ destination}
        \includegraphics[scale=0.5]{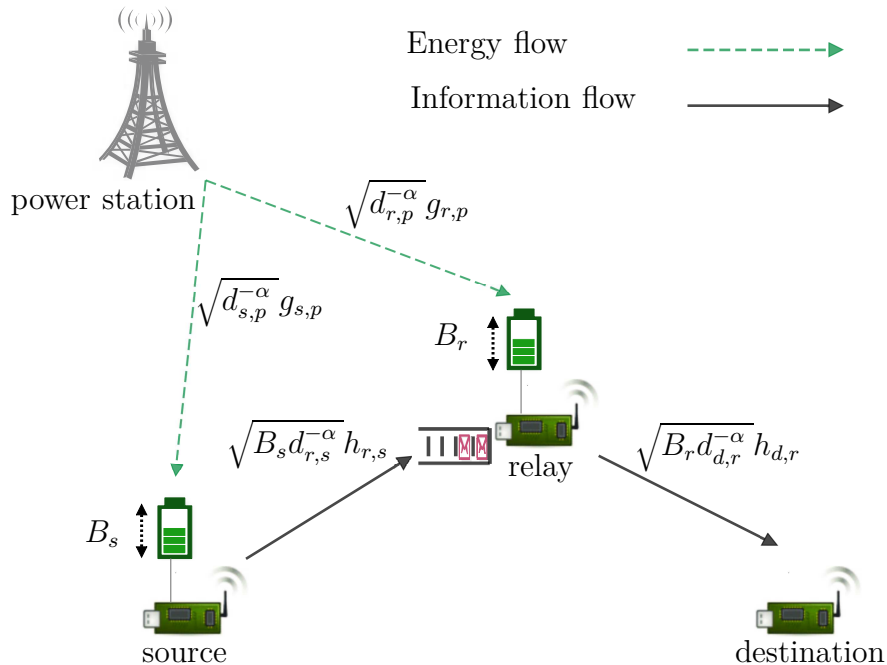}
        \caption{Wireless powered cooperative network.}
        \label{fig:sysmodel}
\end{figure}

\section{System Model}\label{sec2sysmodel}
We consider an age-sensitive wireless powered cooperative communications system, as depicted in Fig.~\ref{fig:sysmodel}. The system includes a power station, a single source, a relay, and a destination node. It is assumed that the power station node is connected to a power grid with unlimited power supply and continuously broadcasts RF energy signals with the power of $P_t$ to charge capacitors of source and relay. Both source and relay nodes capture and store the harvested energy in their finite-sized capacitors with sizes of $B_s$ and $B_r$, respectively. Moreover, it is also assumed that the source always has data to transmit, and only waits for energy arrival to charge its capacitor. The whole transmission time is slotted with an equivalent time unit, where $k\in\{1, \dots, K, \dots, \infty\}$ denotes the index of the time slots. 


Considering a greedy policy, which sends a new status update whenever there is sufficient energy,  the two nodes can transmit data at the beginning of the next time slot provided that their capacitors are fully charged. Thus, they use all the available energy for data transmission with maximum power. It is also assumed that the energy transmissions and data communications are performed over orthogonal channels, i.e. different frequency bands are used for energy and data transmission, thus during data transmission, the source and relay nodes are able to harvest energy from the power station. Besides, all the channels are modeled by a block flat fading Rayleigh one.

We consider both DF and AF relaying schemes. For the case of two-hop DF scheme, it can be assumed that there are two independent links; firstly, the source transmits data to the relay node and if it can be decoded successfully at the relay, then the relay re-transmits the message to the destination provided that it has a fully charged capacitor; otherwise, it buffers the data and waits for harvesting more energy. The buffer-aided relay has an infinite sized buffer, and is equipped with a data queue due to its service time. For the case of the AF method, the source broadcasts the updated signal while the relay node works as a signal amplifier. For this case, the source transmits when both of the capacitors at the source and the relay are fully charged.

Before analysing the age metrics for the proposed WPT cooperative sensor network in Section \ref{sec3AOI}, in the following subsections, we firstly introduce the block fading channel model and the respective channel coefficients. Then, the conventional linear WPT model is discussed and the probability that a capacitor becomes fully charged in a given time period is derived. Moreover, probability of successful data transmissions over communication links are presented. Finally, average AoI and PAoI concepts are presented. 



 \subsection{Channel Model}
As depicted on Fig.~\ref{fig:sysmodel},  path loss between the two-node pair $(i,j)$ is denoted by $d_{i,j}^{-\alpha}$ for $(i,j)\in\{(s,p), (r,p), (r,s), (d,r)\}$, where $\alpha$ denotes the path loss exponent. Besides, the additive noise at each node is modeled by a zero-mean circularly symmetric complex Gaussian random variable with  variance \textcolor{black}{$\sigma^{2}$}. Let us define $\{g_{l,p}^k\}$ for $l\in\{s,r\}$ to denote the small-scale fading coefficients of the energy-transfer links between the (source--power station) and (relay--power station) pairs at the $k$-th time slot. Also, $\{h_{r,s}^k, \textcolor{black}{h_{d,r}^k}\}$ indicate the small-scale fading coefficients for  (relay--source) and (destination--relay) pairs at the $k$-th time slot, respectively. All the channel coefficients $\{g_{s,p}^k, g_{r,p}^k, h_{r,s}^k, h_{d,r}^k\}$ are modeled by independent and identically distributed (i.i.d) zero-mean and unit variance complex Gaussian random variables, i.e.  $\mathcal{CN}(0,1)$. Therefore, all the links are assumed to be block fading Rayleigh ones, and the channel parameters are stationary within the one-time slot and change independently from slot to slot. 

\subsection{Wireless Power Transfer}
Assume that the power station node sends an energy signal $X_p^k \sim\mathcal{N}(0,P_t)$ at the $k$-th time slot. Thus, the source and the relay nodes receive the following signal
\begin{equation}
    Y_{l,p}^k = \sqrt{d_{l,p}^{-\alpha}} g_{l,p}^kX_p^k+N_{l,p}^k,
\end{equation}
where $ l \in \{s,r \} $ and $N^k_{l,p}\sim \mathcal{CN}\sim(0,\sigma^2)$ denotes the independent additive noise. Therefore, the harvested energy at the $k$-th time slot in the source and the relay nodes are given by  
\begin{equation}
 E_l^k = \frac{\eta P_t {|g_{l,p}^k|}^2}{d_{l,p}^{\alpha}}, 
\end{equation}
where $\eta$ is the energy transfer efficiency. It is worth noting that, we select this conventional linear model,  
as widely used in the energy harvesting literature \cite{kirikidis,uplink}, for the sake of simplicity, and more advanced and practical models could be investigated in future works.

At the $k$-th time slot, data transmission at each node, i.e. the source or the relay, occurs if the \textit{finite}-sized capacitor becomes fully charged. Therefore, an important fact that  affects the AoI is the delay caused by the energy collection period to fully charge the capacitor. This time interval depends on two factors; 1- \emph{size of capacitor} and 2- \textit{randomness} of the amount of harvested energy in each time slot. Thus, the main question is that how long does it take to fully charge the source and the relay capacitors with the size of $B_s$ and $B_r$, respectively. Now, we define $T_s$ to denote a random variable with probability mass function (p.m.f) $P_s(m)$ to model the random waiting time to fully charge the capacitor of the source in $m$ consecutive time slots. Therefore, we have


\begin{align}
P_s(m) &= \mathbb{P}_r\left \{ {\Big(\sum_{i=1}^{m-1}E_s^i < B_{s}\Big) \cap \Big(\sum_{i=1}^{m}E_s^i \geq B_{s}\Big)} \right \} \nonumber 
\\ &= \mathbb{P}_r\left \{ {\Big(\sum_{i=1}^{m-1}\frac{\eta P_t {|g_{s,p}^i|}^2}{d_{s,p}^{\alpha}} < B_{s}\Big) \cap \Big(\sum_{i=1}^{m}\frac{\eta P_t {|g_{s,p}^i|}^2}{d_{s,p}^{\alpha}} \geq B_{s}\Big)} \right \} \nonumber
\\ &= \mathbb{P}_r\left \{ {\Big(\sum_{i=1}^{m-1}|g_{s,p}^i|^2 < B_{s}^\prime\Big) \cap \Big(\sum_{i=1}^{m}|g_{s,p}^i|^2 \geq B_{s}^\prime\Big)} \right \} \nonumber 
\\ &= \mathbb{P}_r\left  \{\Big(Z < B_{s}^\prime\Big) \cap \Big(Z+V \geq B_s^{\prime}\Big)\right \} \nonumber 
\\ &= \int_{z=0}^{B_s^{\prime}}\int_{v=B_s^{\prime}-z}^{\infty}f_Z(z)f_V(v)dvdz \nonumber
\\ &= \int_{z=0}^{B_s^{\prime}}f_Z(z)\Big[1-F_V(B_s^{\prime}-z)\Big]dz \nonumber
\\ &= \frac{(B_s^{\prime})^{m-1} \exp(- B_s^{\prime})}{(m-1)!} \label{eq:source battry size},
\end{align}
where $B_s^{\prime} \triangleq  \frac{d_{s,p}^{\alpha} B_s}{\eta P_t }$, $Z \triangleq \sum_{i=1}^{m-1}{|g_{s,p}^i|}^2\sim\textnormal{Erlang}({m-1}, 1)$, and $V\triangleq{|g_{s,p}^m|}^2 \sim \textnormal{exp}(1)$. 

Similarly at the relay,  p.m.f. of the relay's capacitor becoming fully charged exactly in $m$ time slots becomes
\begin{equation}
T_r \sim P_r(m)=\frac{(B_r^{\prime})^{m-1} \exp(- B_r^{\prime})}{(m-1)!},\label{eq:relay battry size}
\end{equation}
where $B_r^{\prime} \triangleq  \frac{d_{r,p}^{\alpha} B_r}{\eta P_t }$. In the following Proposition, first and second moments of $T_s$ and $T_r$, are evaluated.
\begin{proposition}
\label{thm:FPA}
For $T_l \sim P_l(m)~$ with $l \in \{s,r\}$, given in (\ref{eq:source battry size}) and (\ref{eq:relay battry size}), we have
\begin{subequations}
\begin{align}
\mathbb{E}[T_l] &= \exp(-B_l^{\prime})\sum_{m=1}^{\infty}\frac{m{(B_l^{\prime})}^{m-1}}{(m-1)!}  \nonumber
\\ &= 1+ B_l^{\prime},
  \\\mathbb{E}[T_l^2] &= \exp(-B_l^{\prime})\sum_{m=1}^{\infty}\frac{m^2{(B_l^{\prime})}^{m-1}}{(m-1)!} \nonumber
\\ &= 1+3B_l^{\prime}+(B_l^{\prime})^2.
\end{align}
\end{subequations}
\end{proposition}
The results of Proposition~\ref{thm:FPA} are used in the next Section for calculating the average PAoI.
\subsection{Successful Data Transmission}
In the following, we briefly present cooperative data transmission for the DF and AF relaying schemes.
\subsubsection{Cooperative DF scheme}
The source and the relay start transmitting information right after their capacitors become fully charged. The signals received at the relay and the destination at the $k$-th time slot are respectively given by
\begin{subequations}
\begin{equation}
Y_{r,s}^k = \sqrt{B_sd_{r,s}^{-\alpha}} h_{r,s}^kX_s^k+N_{r,s}^k,
\end{equation}
\begin{equation}
Y_{d,r}^k = \sqrt{\textcolor{black}{{B^*_r}}d_{d,r}^{-\alpha}} h_{d,r}^kX_r^k+N_{d,r}^k,
\end{equation}
\end{subequations}
where $X_s$ and $X_r$ denote the unit energy transmitted information signal from source and relay. The $Y_{r,s}$ and $Y_{d,r}$ denote the signals received at the relay from source and destination from relay, respectively. \textcolor{black}{The $B_s$ and $B^*_r$ are the average transmitted energy, at $k$-th time slot from source and relay, respectively. For the DF relay, we define $B^*_r=B_r - P_{c}$, where $P_c$ indicates the average consumed processing power at the relay and discussed in the following. Besides, for the AF relay we have $B^*_r=B_r$. By assuming all the time slots have unit duration, $B_s$ and $B^*_r$ are interpreted as the transmission power.} The $N^k_{r,s}$ and $N^k_{d,r}$ respectively denote the independent additive noise in the (relay--source) and the (destination--relay) links with distribution $\mathcal{CN}\sim(0,\sigma^2)$. Therefore, the signal-to-noise ratio (SNR) of the (relay--source) and (destination--relay) links at the $k$-th time slot transmission are respectively given by
\begin{subequations}
\begin{equation}
\gamma_{r,s}^k =\frac{{B_s|h_{r,s}^k|^2}}{d_{r,s}^\alpha\sigma^2}, 
\end{equation}
\begin{equation}
\gamma_{d,r}^k =\frac{{B^*_r|h_{d,r}^k|^2}}{d_{d,r}^\alpha\sigma^2}. 
\end{equation}
\end{subequations}
Bedsides, as $h_{r,s}^k, h_{d,r}^k $ are modeled by complex Gaussian random variables, $ |h_{r,s}^k|^2, |h_{d,r}^k|^2 $ are exponentially distributed with unit mean.

To perfectly recover the information at the relay or the designation, the SNR must be greater than $\gamma_{th}$. Accordingly, probability of successful data transmission over the source--relay and the relay--destination channels, becomes
\begin{subequations}\label{eq:Psuc}
\begin{align}
P_{suc,s}&=\mathbb{P}_r \left \{\gamma_{r,s}\geq \gamma_{th}\right \}\nonumber
\\&= \exp\left(\frac{-\sigma^2 \gamma_{th} }{B_sd_{r,s}^{- \alpha}}\right), 
\label{eq:successs}
\\P_{suc,r}&=\mathbb{P}_r \left \{\gamma_{d,r}\geq \gamma_{th} \right \}\nonumber
\\&= \exp\left(\frac{-\sigma^2 \gamma_{th} }{B^*_rd_{d,r}^{- \alpha}}\right). 
\label{eq:successr}
\end{align}
\end{subequations}

Since with probability $(1-P_{suc,r})$ the relay may not be able to successfully recover the information, the source re-transmits the signal till be recovered successfully at the relay. It is assumed that whenever the relay applies the decoding processing, it consumes $C_p$ units of power, which is called the \textit{processing cost}. Therefore, the expected number of the re-transmission attempts from the source node, i.e.  $n_{re-trans}^s$, becomes
\begin{align} 
n_{re-trans}^s &\triangleq \mathbb{E}[I] = \sum_{i=0}^{\infty}iP_{suc,s}{(1-P_{suc,s})}^{(i-1)}  \nonumber
\\ &= 
\frac{1}{P_{suc,s}}.
\end{align}
Thus, the average consumed processing power at the relay becomes $P_{c} = C_p \times n_{re-trans}^s = \frac{C_p}{P_{suc,s}}$, and  the average transmission power becomes $B^*_r= B_r - P_{c}$.

\subsubsection{Cooperative AF scheme}
For the AF relaying scheme, the relay just amplifies the received noisy signal and re-transmits it to the destination. Therefore, there is no need to assume the  processing cost at the relay.  Hence to transmit with maximum power $B_r$ at the relay, we have 
\begin{align}
    Y_{d,r}^k &=  \frac{\sqrt{B_{r} d_{d,r}^{-\alpha}}}{\sqrt{B_{s} d_{r,s}^{-\alpha}{|h_{r,s}^k|}^2+\sigma^2}}h_{d,r}^k\sqrt{B_{s} d_{r,s}^{-\alpha}}h_{r,s}^k X_s^k+{N_{d,r}^{\prime k}}. 
\end{align}
The \textcolor{black}{${N_{d,r}^{\prime k}}$} is the effective noise at the destination, and is defined as
\begin{equation}
    {N_{d,r}^{\prime k}}= \frac{\sqrt{B_{r}d_{d,r}^{-\alpha}}}{\sqrt{B_{s} d_{r,s}^{-\alpha}{|h_{r,s}^k|}^2+\sigma^2}}h_{d,r}^kN_{r,s}^k+N_{d,r}^k, 
\end{equation}
Similar to \cite{Chen2010}, the end-to-end SNR is derived and to have a successful data transmission at the $k$-th time slot, we have
\begin{align}
    P_{suc, AF} &= \textcolor{black}{\mathbb{P}_r \left \{ \gamma_{d,s} \geq \gamma_{th} \right \} } \nonumber
    \\&= \exp\left(-\gamma_{th}\sigma^2\Big(\frac{1}{B_{s}d_{r,s}^{-\alpha}}+\frac{1}{B_{r}d_{d,r}^{-\alpha}}\Big)\right) \sqrt{\beta} K_1(\sqrt{\beta}),
    \label{eq:af success}
\end{align}
where $\beta = \frac{4\sigma^4 \Big(\gamma_{th}^2+\gamma_{th} \Big)}{B_{s}B_{r} d_{r,s}^{-\alpha} d_{d,r}^{-\alpha}}$, and $K_1(.)$ is the first order modified Bessel function of the second kind. 

\subsection{Age of Information}
For a conventional \textit{one-hop} network, the AoI at the destination at time slot $n$, i.e. $\Delta(n)$, is defined as the  
\begin{equation}
\Delta(n) = n - U(n),
\end{equation}
where $U(n)$ is the time slot at which the most recent update has been transmitted from the source node. Fig.~\ref{fig:age} presents a sample of the age evolution for the conventional one-hop network with initial age $\Delta^0 $. The $n^k$ and $n^{k+1}$ represent the time slots of two succeeding updates at the destination node, and if the destination could successfully receive the message, the AoI is reset to one. Also, $X^k$ denotes the $k$-th interarrival time, which is the time between the two successful receptions of the $k$-th and the $(k+1)$-th updates. Therefore, we have
 \begin{equation}\label{eq:Xk}
X^k =\sum_{i=1}^{m^k} T^i,
\end{equation}
where $T^i$ denotes the time interval between two data transmissions or capacitor recharge, and $m^k$ is the realization of a discrete random variable $M$ between the $k$-th and $(k+1)$-th successful packet reception. It is worth mentioning that $X^k$ is a stationary random process; hence we use $\mathbb{E}[X]$ as the expected value of $X^k$ for an arbitrary $k$.


As depicted in Fig.~\ref{fig:age}, peak values of the saw-tooth curve indicate the maximum value of AoI right before receiving a fresh update , i.e. $A^k=X^k ~\text{for}~ k=1,2,\cdots$. If one may not expect the age exceeds a predetermined threshold, then it is more useful to consider peak of the age. In other words, receiving a packet after a long period may not carry any useful information, thus peak of age becomes more important and provides information about the the worst case of the age. Note that the peak of age is a discrete stochastic process that takes values at the time instances. Besides the importance of analysing the peak metric, it is also easier to derive \textit{closed-form} of average PAoI than the average AoI \cite{costa2014age , costa2016age}.

 \begin{figure}[!t]
        \centering
         \psfrag{delta} {\hspace*{-1.7mm} $\Delta^0$}
         \psfrag{t} {\hspace*{-13mm} \ $T^{1}$}
         \psfrag{t1} {\hspace*{-6mm} \ $T^{m^k}$}
         \psfrag{ak} {\hspace*{-6mm} \ $A^{k}$}
         \psfrag{nk} {\hspace*{-2mm} \ $n^{k}$}
         \psfrag{n1} {\hspace*{-3mm} \ $n^{1}$}
         \psfrag{q0} {\hspace*{-2mm} \ $Q^{0}$}
         \psfrag{nk1} {\hspace*{-2.5mm} \ $n^{k+1}$}
         \psfrag{x} {\hspace*{-3mm} \ $X^{k}$}
         \psfrag{x+1} {\hspace*{-3mm} $X^{k}$}
         \psfrag{q} {\hspace*{-3mm} \ $Q^{k}$}
         \psfrag{q1} {\hspace*{-3mm} \ $Q^{k+1}$}
        \psfrag{AoI} {\hspace*{-7mm} $\Delta(n)$}
        \psfrag{n} {\hspace*{-2.5mm} $\wr \wr$}
        \psfrag{Time} {\hspace*{2mm}$n$}
        \includegraphics[scale=0.35]{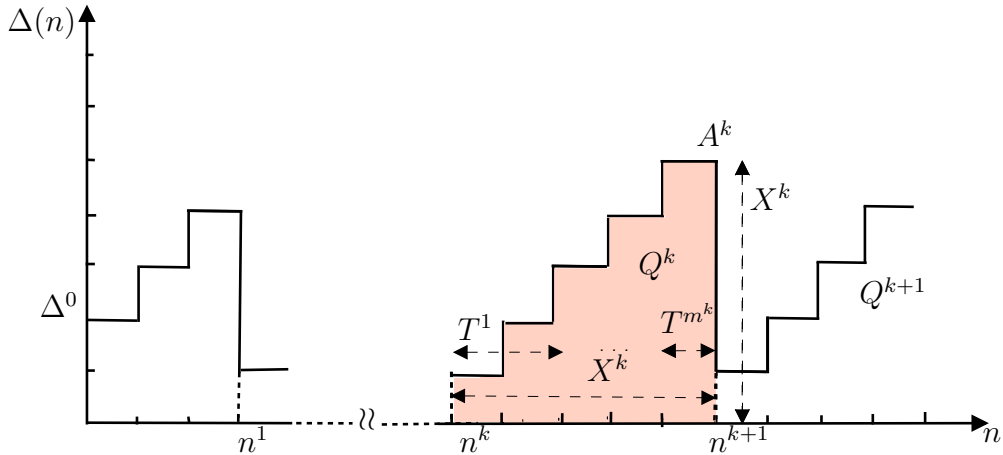}
        \caption{\textcolor{black}{A sample path of AoI $\Delta(n)$ and PAoI $A^k$.}}
        \label{fig:age}
\end{figure}

To analyse the proposed system, we utilize two metrics; the average AoI and average PAoI. The average AoI is defined as the ratio of the expected enclosed surface by the $\Delta(n)$ and time axis over the expected length of time between the two successful transmissions. As a result, the average AoI for a period of $N$ time slots, is given by
\begin{equation}
\Delta^{N} = \frac{1}{N}\sum_{n = 1}^{N}\Delta(n) = 
\frac{1}{N}\sum_{k = 1}^{K}Q^k=\frac{K}{N}\frac{1}{K}\sum_{k = 1}^{K}Q^k,
\end{equation}
where $Q_k$ denotes the $k$-th status update of the area under $\Delta(n)$. Since $\lim\limits_{N\to\infty}\frac{K}{N}=\frac{1}{\mathbb{E}[X]}$ is the steady state rate of updates generation and $\frac{1}{K}\sum_{k = 1}^{K}Q^k$  converges to $\mathbb{E}[Q]$ when $K$ goes to infinity, the time average of $\Delta^{N}$ tends to the ensemble average age as
\begin{equation}\label{eq:AoI_def}
\Delta = \lim\limits_{N\to\infty}\Delta^N=\frac{\mathbb{E}[Q]}{\mathbb{E}[X]}.
\end{equation}
 Besides the average PAoI is defined as
\begin{equation}\label{eq:peak}
A = \mathbb{E}[X]. 
\end{equation}


\section{Average Peak Age of Information Analysis}\label{sec3AOI}

From information-theoretic points of view, it is well-known that utilizing a cooperative relay can improve the system performance, especially when the direct link between the source and the destination is not good enough \cite{el2011network}. It is worth mentioning that in the proposed system model, the relay may improve the throughput or equivalently increase the probability of successful transmission but at the cost of increasing the delay overhead due to charging its capacitor. Therefore, it is not obvious how relaying schemes affect AoI-based metrics. In the following, we analyse the average AoI and PAoI for the proposed WPT-based relay system and compare the results with the direct data transmission scheme with \textit{no} relay, as well. We show that utilizing a relay can improve average PAoI by reducing the need for data re-transmission, especially for a larger value of $\gamma_{th}$.



\subsection{Decode-and-Forward  Relaying}
For the case of the DF relaying scheme, the packets which successfully received at the relay are stored in the buffer and wait for the re-transmission to the destination. Thus, the packets will be served in the queue, and the relay transmits one packet to the destination as soon as its capacitor is fully charged. Fig.~\ref{fig:age2} depicts an example of AoI and PAoI evolution for this scenario. 

Here, we summarize main relevant random variables that affect the end-to-end data transmissions and the age evolution;
\begin{itemize}
    \item As the source must wait to charge its capacitor, this waiting time is modeled by the random variable $T_s$ with the given p.m.f in (\ref{eq:source battry size}). $T_s^i$ denotes the random time period between the two consecutive source's capacitor recharges.
    
    \item The source re-transmits the signal to be successfully received at the relay. Therefore, the interarrival time of the source, which is the time between the two successful receptions at the relay, is modeled by $X_s^k=\sum_{i=1}^{m^k} T_s^i$, where $m^k$ is the realization of a discrete random variable $M$ between the $k$-th and $(k+1)$-th successful packet reception.
        
    \item In the relay, the packets wait in a queue to be transmitted. Thus, the waiting time of the packets in the queue is defined as $W^k$ at the $k$-th time slot, and if the queue is empty, $W^k$ becomes zero.
    
    \item The waiting time to charge the capacitor of the relay is modeled by $T_r$ with the given p.m.f in (\ref{eq:relay battry size}), and $T_r^i$ denotes the time (in time slots) between two consecutive relay capacitor recharges.
    
    \item Service time of the system, which is modeled by $X_r^k=\sum_{i=1}^{m^k} T_r^i$, indicates the time slots between the two successful receptions from the relay at the destination node.
    
    \item Finally, the total waiting time in the system, also called the system time, is $Y^k = X_r^k + W^k$ in $k$-th time slot.
\end{itemize}
It is worth noting that for the case of the DF WPT-based network, we can model our system with a queue such that its interarrival time denoted by $X_s^k$, and its service time is $X_r^k$. Since $X_s^k$ and $X_r^k$ do not have renowned distributions, the corresponding queue becomes as a G/G/1 one.

 \begin{figure}[!t]
        \centering
         \psfrag{delta} {\hspace*{-1.7mm} $\Delta^0$}
         \psfrag{t} {\hspace*{-4mm} \ $Y^1$}
         \psfrag{a1} {\hspace*{-4mm} \ $A^0$}
         \psfrag{a} {\hspace*{-4mm} \ $A^1$}
         \psfrag{nk} {\hspace*{-2mm} \ $n^1$}
         \psfrag{nk1} {\hspace*{-2.5mm} \ $n^2$}
         \psfrag{x} {\hspace*{-3mm} \ $X_s^1$}
         \psfrag{q0} {\hspace*{-3mm} \ $Q^0$}
         \psfrag{q1} {\hspace*{-3mm} \ $Q^1$}
        \psfrag{q2} {\hspace*{-3mm} \ $Q^2$}
        \psfrag{AoI} {\hspace*{-7mm} $\Delta(n)$}
        \psfrag{Time} { $n$}
        \includegraphics[scale=0.32]{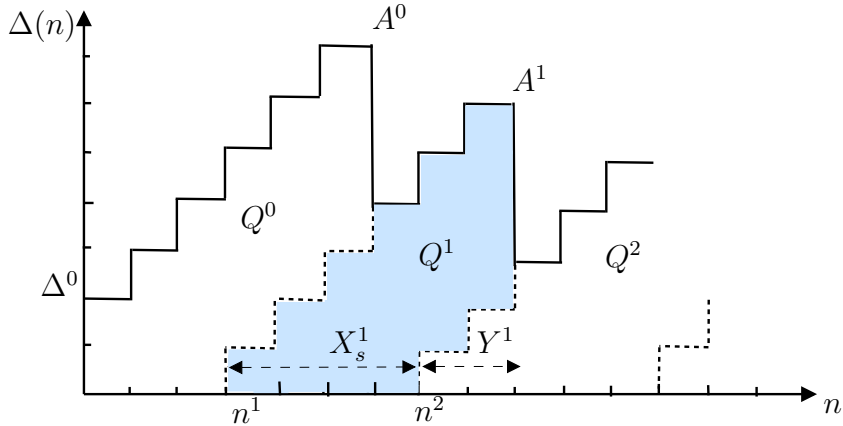}
        \caption{Example of AoI, $\Delta(n)$, and PAoI for the proposed WPT-based DF relaying system.}
        \label{fig:age2}
\end{figure}

According to Fig.~\ref{fig:age2} and by using the previous definitions of  average AoI given in (\ref{eq:AoI_def}), the average AoI could be calculated. Moreover, for the end-to-end \emph{average} PAoI of the cooperative DF relay, we have
\begin{equation}\label{eq:df}
 A_{DF}=\mathbb{E}[X_s + Y] = \mathbb{E}[X_s +X_r +W]. 
\end{equation}

The average PAoI and average AoI of the cooperative DF relay system are presented in the following.

\begin{theorem}\label{thm:FTA}
For the average PAoI of the DF-cooperative system, we have 
\begin{equation}
A_{DF} \leq \mathbb{E}[X_s]+\mathbb{E}[X_r]+ \frac{\mathbb{E}[X_s^2]+\mathbb{E}[X_r^2] -\Big (\mathbb{E}^2[X_s]+\mathbb{E}^2[X_r]\Big )}{2 \Big( \mathbb{E}[X_s]-\mathbb{E}[X_r]\Big)}, \label{eq:final df}
\end{equation}
provided that $ \mathbb{E}[X_r] < \mathbb{E}[X_s]$ to have a stable queue.
\end{theorem}
\begin{proof}
See Appendix A.
\end{proof}

\begin{proposition}
\label{pro:2}
The average AoI for the DF-cooperative system is
\begin{equation}
\Delta_{DF} =\frac{\mathbb{E}[Q]}{\mathbb{E}[X_{s}]}= \frac{\mathbb{E}[X_s^2]/2 + \mathbb{E}[X_s]\mathbb{E}[X_r] +\mathbb{E}[X_sW]}{\mathbb{E}[X_s]}.
\end{equation}
Note that computing $\mathbb{E}[X_sW]$ is very difficult because we cannot evaluate the queue's waiting time distribution due to the unknown distribution of interarrival time $X_s$ and service time $X_r$. Thus,  it is not possible to present a closed-form expression and  must be evaluated numerically.
\end{proposition}

\subsection{Amplify-and-Forward Relaying}

In the AF relaying scheme, since the relay node is \emph{not} capable of decoding and storing the message, the end-to-end data transmission can be modeled similarly to a single-hop network. The difference of the AF-based cooperative system and the direct transmission is that for the case of AF relaying the \textit{total} system waiting time for charging the capacitors is equal to the maximum of two random waiting times at the source at the relay to fully charge \textit{both} of the capacitors, that is $T_{AF} := \max(T_s,T_r)$. In summary, we have the following random variables to model the system.
\begin{itemize}
    \item Random waiting to fully charge both capacitors of the source and relay nodes is modeled by the random variable $T_{AF}$. Besides, $T_{AF}^i$ denotes the random time period between the two consecutive capacitor charges.
    
    \item Due to the randomness of the end-to-end channel, the source may re-transmit the signal to successfully recover information at the destination. The time between the two successful reception at the destination is modeled by $X_{AF}^k=\sum_{i=1}^{m^k} T_{AF}^i$ where $m^k$ is the realization of a discrete random variable $M$ between the $k$-th and $(k+1)$-th successful packet reception.

\end{itemize}


Since the age evolution of the AF scenario is similar to the one-hop network, by use of (\ref{eq:AoI_def})--(\ref{eq:peak}), and defining the parameters as $T^i=T_{AF}^i$ and $X^k=X_{AF}^k$, average AoI and average PAoI for the AF relaying scheme becomes similar to the one depicted in Fig.~\ref{fig:age}. Therefore, we have the following results on the average AoI and average PAoI for AF scenario. 
\begin{corollary}
\label{thm:FTB}
For the AF-based cooperative system, the average AoI is given by  
\begin{equation}
\Delta_{AF} = \frac{\mathbb{E}[Q]}{\mathbb{E}[X_{AF}]}=\frac{\mathbb{E}[ \frac{X_{AF}\times (X_{AF}+1)}{2}]}{\mathbb{E}[X_{AF}]}=\frac{1}{2}
\Big(\frac{\mathbb{E}[X_{AF}^2]}{\mathbb{E}[X_{AF}]} +1 \Big)
.
\end{equation}
Therefore, we can compute  the \emph{exact} average AoI for the AF scheme. Besides, the average PAoI becomes
\begin{equation}
A_{AF}=\mathbb{E}[X_{AF}].
\end{equation}
\end{corollary}

Thus, to compute the AoI metrics, we just need to compute $\mathbb{E}[X_{AF}]$ and $\mathbb{E}[X^2_{AF}]$ respectively similar to (\ref{eq:subeq1}) and (\ref{eq:subeq2}) by replacing $l=AF$. Therefore, we need to compute $\mathbb{E}[T_{AF}]$ and $\mathbb{E}[T^2_{AF}]$ which are given in the following Proposition.

\begin{proposition}\label{props:3}For the random variable $T_{AF}^i = \max(T_s^i,T_r^i)$, we have  
\begin{subequations}
\begin{align}
\mathbb{E}[T_{AF}]&= 1+ \sum_{i=1}^{\infty} \Big(1-\frac{\Gamma(i,B_s^\prime)\Gamma(i,B_r^\prime)}{\Gamma(i)\Gamma(i)}\Big), \label{eq:etlast}
\end{align}
\begin{align}
    \mathbb{E}[T_{AF}^2]&=  2\sum_{i=1}^{\infty} i \Big(1-\frac{\Gamma(i,B_s^\prime)\Gamma(i,B_r^\prime)}{\Gamma(i)\Gamma(i)}\Big),\label{eq:et2last}
\end{align}
\end{subequations}
where $\Gamma(i,j)$ is the incomplete Gamma-distributed  cumulative functions with shape $i$ and scale $j$, and $\Gamma(i)$ denotes the complete Gamma function, as well.
\end{proposition}
\begin{proof} See Appendix B.
\end{proof}

\section{Special Cases}\label{sec4S}
Since our proposed system model is general, we showed that the corresponding queuing type, especially for the DF relaying, becomes an arbitrary and it is not straight forward to derive the exact value of the average PAoI. In this section, we briefly investigate special cases of the proposed general system model and show that the results of Theorem 1 and Corollary 1 are reduced to include age analysis of special cases, for instance full-power nodes, one-shot successful data transmission and etc, and in some cases types of the queue are well-known, and the exact mean waiting times could be derived in some cases as well. In Table~\ref{tab:my-table}, we investigate all special cases for the DF scheme; Some queues are always stable, some of them are not, and others are stable subject to specific conditions. For the AF scheme, we later discuss four special cases as well.

For the DF-based cooperative system, the age evolution depends on two main factors; 1- \emph{capacitor charging} at the source/relay, and 2- \textit{successful transmission} over the (source--relay)/(relay--destination) link, which means if $k$ re-transmissions occurs, ($k-1$) consecutive transmissions were unsuccessful while the $k$-th transmission was successful. Table~\ref{tab:my-table} summarizes the result for each of the DF scheme cases; the queuing type is clarified, the stability condition is provided (some cases are not stable, some are always stable, and some may need specific conditions to be stable), and the average PAoI is presented. In Table~\ref{tab:my-table}, for $ l \in \{s,r\}$, the following terminologies are used;
\begin{table}[!t]
\caption{Special cases for the DF cooperative scheme.}
\label{tab:my-table}
\resizebox{\textwidth}{!}{%
\begin{tabular}{|cc|cc|c|c|c|}
\hline
\rowcolor[HTML]{CBCEFB}
\multicolumn{2}{|c|}{\textbf{At source}}& \multicolumn{2}{c|}{\textbf{At relay}}& 
\textbf{Type of}& 
\textbf{Stability condition}&
\textbf{Average PAoI}\\
\rowcolor[HTML]{CBCEFB}
capacitor & $P_{suc,s}$& capacitor & $P_{suc,r}$&\textbf{queue}&$\rho < 1$& $A_{DF} =\mathbb{E}[X_s]+\mathbb{E}[X_r]+\mathbb{E}[W]$\\\hline\hline

\cellcolor[HTML]{FFEFDB}full     & \cellcolor[HTML]{FFEFDB}1      & \cellcolor[HTML]{FFEFDB}full     & \cellcolor[HTML]{FFEFDB}1      & D/D/1                        & \cellcolor[HTML]{B8E9A8} \checkmark  & 2 \\ \hline
\cellcolor[HTML]{FFEFDB}full     & \cellcolor[HTML]{FFEFDB}1      & \cellcolor[HTML]{FFE1BB}full     & \cellcolor[HTML]{FFE1BB}random & D/Geo/1                      & \cellcolor[HTML]{FF9B9B}$P_{suc,r} > 1$  \textreferencemark  & The queue is unstable  \\ \hline
\cellcolor[HTML]{FFEFDB}full     & \cellcolor[HTML]{FFEFDB}1      & \cellcolor[HTML]{FFD094}charging & \cellcolor[HTML]{FFD094}1      & D/P/1                        & \cellcolor[HTML]{FF9B9B}$B^\prime_r < 0$  \textreferencemark      & The queue is unstable  \\ \hline

\cellcolor[HTML]{FFEFDB}full     & \cellcolor[HTML]{FFEFDB}1      & \cellcolor[HTML]{FFB556}charging & \cellcolor[HTML]{FFB556}random & D/G/1                        & \cellcolor[HTML]{FF9B9B}$\mathbb{E}[X_r] < 1$ \textreferencemark & The queue is unstable\\\hline

\cellcolor[HTML]{FFE1BB}full     & \cellcolor[HTML]{FFE1BB}random & \cellcolor[HTML]{FFEFDB}full     & \cellcolor[HTML]{FFEFDB}1      & Geo/D/1                      & \cellcolor[HTML]{B8E9A8}$ P_{suc,s} < 1$  \checkmark & $ \frac{3}{2P_{suc,s}}+1$ \\ \hline

\cellcolor[HTML]{FFE1BB}full     & \cellcolor[HTML]{FFE1BB}random & \cellcolor[HTML]{FFE1BB}full     & \cellcolor[HTML]{FFE1BB}random & Geo/Geo/1                    & $P_{suc,s} <P_{suc,r} $                & $\frac{P_{suc,r}^2(3-P_{suc,s})-P_{suc,s}^2(1+P_{suc,r})}{2P_{suc,s}P_{suc,r}(P_{suc,r}-P_{suc,s})}$                                                     \\ \hline
\cellcolor[HTML]{FFE1BB}full     & \cellcolor[HTML]{FFE1BB}random & \cellcolor[HTML]{FFD094}charging & \cellcolor[HTML]{FFD094}1      & Geo/P/1                      & $  1+B^\prime_r < \frac{1}{P_{suc,s}}$     & $\frac{1}{P_{suc,s}}+(1+B^\prime_r)+ \frac{\frac{1-P_{suc,s}}{P^2_{suc,s}}+B^\prime_r}{2 \Big( \frac{1}{P_{suc,s}}-(1+B^\prime_r)\Big)}$                               \\ \hline

\cellcolor[HTML]{FFE1BB}full     & \cellcolor[HTML]{FFE1BB}random & \cellcolor[HTML]{FFB556}charging & \cellcolor[HTML]{FFB556}random & Geo/G/1                      & $\mathbb{E}[X_r] < \frac{1}{P_{suc,s}}$    & $\frac{1}{P_{suc,s}}+\mathbb{E}[X_r] + \frac{\frac{1-P_{suc,s}}{P^2_{suc,s}}+\mathbb{E}[X_r^2]-\mathbb{E}^2[X_r] }{2 \Big( \frac{1}{P_{suc,s}}-\mathbb{E}[X_r] \Big)}$ \\ \hline
\cellcolor[HTML]{FFD094}charging & \cellcolor[HTML]{FFD094}1      & \cellcolor[HTML]{FFEFDB}full     & \cellcolor[HTML]{FFEFDB}1      & P/D/1                        & \cellcolor[HTML]{B8E9A8} $B^\prime_s > 0$ \checkmark     & $\frac{5}{2} + B^\prime_s$\\ \hline

\cellcolor[HTML]{FFD094}charging & \cellcolor[HTML]{FFD094}1      & \cellcolor[HTML]{FFE1BB}full     & \cellcolor[HTML]{FFE1BB}random & P/Geo/1                      & $ \frac{1}{P_{suc,r}} < 1+B^\prime_s$      & $(1+B^\prime_s)+\frac{1}{P_{suc,r}} + \frac{B^\prime_s + \frac{1-P_{suc,r}}{P^2_{suc,r}}}{2 \Big((1+B^\prime_s)- \frac{1}{P_{suc,r}} \Big)}$                           \\ \hline
\cellcolor[HTML]{FFD094}charging & \cellcolor[HTML]{FFD094}1      & \cellcolor[HTML]{FFD094}charging & \cellcolor[HTML]{FFD094}1      & P/P/1                        & $ B^\prime_r < B^\prime_s$                     & $ \frac{B^\prime_s(5+2B^\prime_s) - B^\prime_r(3+2B^\prime_r)}{2 \Big(B^\prime_s-B^\prime_r \Big)}$                                                                                       \\ \hline
\cellcolor[HTML]{FFD094}charging & \cellcolor[HTML]{FFD094}1      & \cellcolor[HTML]{FFB556}charging & \cellcolor[HTML]{FFB556}random & P/G/1                        & $\mathbb{E}[X_r] < 1+ B^\prime_s$              & $(1+B^\prime_s)+\mathbb{E}[X_r]+ \frac{ B^\prime_s +\mathbb{E}[X_r^2] -\mathbb{E}^2[X_r] }{2 \Big( (1+B^\prime_s)-\mathbb{E}[X_r]\Big)}$                                                 \\ \hline
\cellcolor[HTML]{FFB556}charging & \cellcolor[HTML]{FFB556}random & \cellcolor[HTML]{FFEFDB}full     & \cellcolor[HTML]{FFEFDB}1      & G/D/1                        & \cellcolor[HTML]{B8E9A8}$ \mathbb{E}[X_s]>1$ \checkmark & $\mathbb{E}[X_s]+1+ \frac{\mathbb{E}[X_s^2] -\mathbb{E}^2[X_s] + 1}{2 \Big( \mathbb{E}[X_s]-1\Big)}$                                                                                     \\ \hline
\cellcolor[HTML]{FFB556}charging & \cellcolor[HTML]{FFB556}random & \cellcolor[HTML]{FFE1BB}full     & \cellcolor[HTML]{FFE1BB}random & G/Geo/1                      & $ \frac{1}{P_{suc,r}} < \mathbb{E}[X_s]$   & $\mathbb{E}[X_s]+\frac{1}{P_{suc,r}}+ \frac{\mathbb{E}[X_s^2]-\mathbb{E}^2[X_s] + \frac{1-P_{suc,r}}{P^2_{suc,r}}}{2 \Big( \mathbb{E}[X_s]-\frac{1}{P_{suc,r}}\Big)}$  \\ \hline
\cellcolor[HTML]{FFB556}charging & \cellcolor[HTML]{FFB556}random & \cellcolor[HTML]{FFD094}charging & \cellcolor[HTML]{FFD094}1      & G/P/1                        & $1 + B^\prime_r < \mathbb{E}[X_s]$             & $\mathbb{E}[X_s]+(1+B^\prime_r)+ \frac{\mathbb{E}[X_s^2]-\mathbb{E}^2[X_s]+ B^\prime_r}{2 \Big( \mathbb{E}[X_s]-(1+B^\prime_r)\Big)}$                                                    \\ \hline
\cellcolor[HTML]{FFB556}charging & \cellcolor[HTML]{FFB556}random & \cellcolor[HTML]{FFB556}charging & \cellcolor[HTML]{FFB556}random & {\color[HTML]{333333} G/G/1} & $\mathbb{E}[X_r] < \mathbb{E}[X_s]$            & $\mathbb{E}[X_s]+\mathbb{E}[X_r]+ \frac{\mathbb{E}[X_s^2]+\mathbb{E}[X_r^2] -\Big (\mathbb{E}^2[X_s]+\mathbb{E}^2[X_r]\Big )}{2 \Big( \mathbb{E}[X_s]-\mathbb{E}[X_r]\Big)}$             \\ \hline
\end{tabular}%
}
\end{table}

\begin{itemize}
    \item capacitor is ``full'' means that there is no need to wait for capturing energy.  This case is equivalent to assuming $P_t \rightarrow \infty$. Also, ``charging'' indicates for capturing energy, so we need to investigate $T_l$. 
    \item $P_{suc,l}$ is ``1'' indicates always receiving all the packets successfully, which is equivalent to assuming $\gamma_{th} \rightarrow 0$, and ``random'' indicates to consider random behavior of the channel and to investigate re-transmission analysis.
    \item Type of queues would vary according to different conditions on capacitors and  $P_{suc,l}$. The following notations used to demonstrate types of the queues
    \begin{itemize}
        \item ``G'' stands for the general case and is used for a situation where the capacitor is charging and $P_{suc,l}$ is random.
    
        \item ``Geo'' stands for the Geometric distribution. If the capacitor is full and $P_{suc,l}$ is random, the distribution becomes the Geometric. 
        
        \item ``P'' refers to the case wherein the capacitor is charging and $P_{suc,l}=1$. To avoid confusion with the ``G'' case, and because of knowing the p.m.f of capacitor charging, we have used ``P''.
        
        \item ``D'' stands for deterministic and is used for the cases with no randomness.
    \end{itemize}
    
    \item Stability of the queue is equivalent to utility $ \rho $ be less than 1. The notation  \textreferencemark indicates that the stability condition never happens, and \checkmark means that the queue is always stable. For instance, for the case of D/Geo/1, it is impossible to have $P_{suc,r} >1$, but for Geo/D/1 condition $P_{suc,s} < 1$ is always true. Also, for queue D/D/1, we have $\rho=1$, and the queue is always stable, too.
    
\end{itemize}

\textcolor{black}{Finally, by using Theorem 1 and applying different conditions, the average PAoI under stability conditions are derived} and summarized in Table.~\ref{tab:my-table} for sixteen different cases of the DF relaying scenario. As we mentioned before, our analysis provides the closed-form upper bound for the mean waiting time. However for the cases of  Geo/Geo/1 and Geo/G/1, which are well-known queues, the exact mean waiting times are derived in \cite{discretequeuebook}. Thus, we have
 \begin{eqnarray}
\mathbb{E}[W] =
\begin{cases}
 \frac{P_{suc,s}(1-P_{suc,s})}{P_{suc,r}(P_{suc,r}-P_{suc,s})}\text{,}\hspace{2.5cm}& \text{for Geo/Geo/1 }\\
 \frac{\mathbb{E}[X_r^2]}{2 ( \frac{1}{P_{suc,s}}-\mathbb{E}[X_r] )}\text{,}\hspace{2.5cm}& \text{for Geo/G/1 }
 \end{cases}
 \label{eg:exact}
\end{eqnarray}

Besides, for the AF relaying system, we have four cases; The general case is presented in Corollary 1. For the case of deterministic one, in which capacitors are full and $P_{suc,AF}=1$, the average PAoI becomes 1. For the situation with fully charged capacitors but random $P_{suc,AF}$, we have $A_{AF}=\frac{1}{P_{suc,AF}}$. Finally, for the capacitor charging case with  $P_{suc,AF}=1$, we have  $A_{AF}=\mathbb{E}[T_{AF}]$, as presented in Corollary 1.


\section{Numerical Results}\label{sec5Ssim}
In this section, we provide numerical results to compare the DF- and AF-based cooperative WPT systems' performance metrics compared to direct transmission in various situations. In the context of information theory, it is convenient to assume parameter dimensionless, i.e. based on random variables. 
Thus without loss of generality and just for simplicity, the noise variance $\sigma^2 = 1$, the energy transfer efficiency $\eta=0.8$, and the processing cost $C_P = 0.01$ are used for the numerical results. Also, it is assumed that the power station is located such that $d_{s,p} = 1$, $d_{r,p} = 1$. 

\begin{figure}[!t]
        \centering
        \psfrag{avarage AoI and PAOI} {\!\!\!\!\!\!\!\!\!\!\!\!\!\!\!\! Average AoI and PAoI}
        \psfrag{pt} {\small \hspace{-6mm}$P_t$}
        \psfrag{theory peak direct}{\footnotesize PAoI direct}
        \psfrag{theory peak AF}{\footnotesize PAoI AF}
        \psfrag{theory peak DF}{\footnotesize PAoI DF}
         \psfrag{simulation peak direct}{\footnotesize PAoI direct simulation}
        \psfrag{simulation peak AF}{\footnotesize PAoI AF simulation}
        \psfrag{simulation peak DF}{\footnotesize PAoI DF simulation}
      \psfrag{theory age direct}{\footnotesize AoI direct}
        \psfrag{theory age AF}{\footnotesize AoI AF}
        \psfrag{theory age DF}{\footnotesize AoI DF}
         \psfrag{simulation age direct}{\footnotesize AoI direct simulation}
        \psfrag{simulation age AF}{\footnotesize AoI AF simulation}
        \psfrag{simulation age DF}{\footnotesize AoI DF simulation}
      
        \includegraphics[scale=0.55]{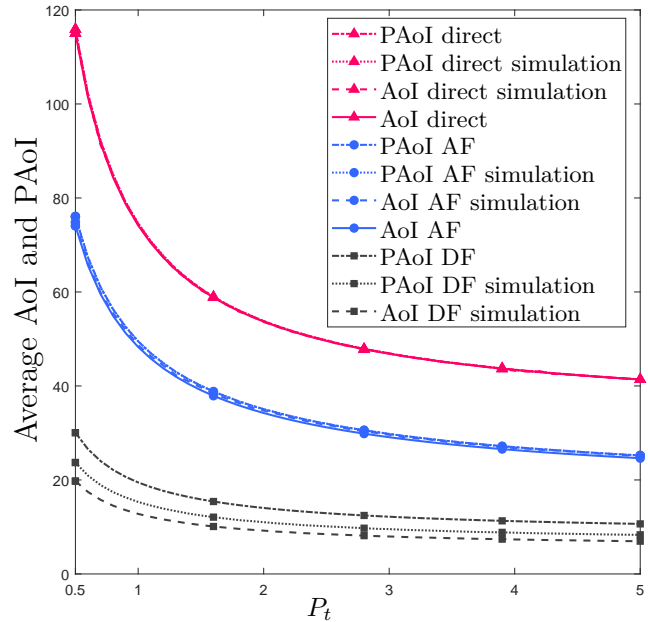}
        \caption{ Average AoI and average PAoI versus transmission power  $P_t$.}

        \label{fig:agepeak}     
\end{figure}

Fig.~\ref{fig:agepeak} depicts the average AoI and average PAoI as a function of the transmission power $P_t$ for three transmission schemes DF, AF, and the direct one for $\alpha = 2$,  $d_{r,s} = 6$, $d_{d,r} = 10$, $\gamma_{th}=16$ [dB] and the capacitor sizes ratio  $\frac{B_s}{B_r}=1$. To verify the derived analytical results, simulation results are also provided for comparison. For the AF relaying and direct transmissions, the analytical results of the two age metrics are perfectly matched with the corresponding simulation results, and also average AoI behavior is similar to the average PAoI. Also for the DF scenario, as discussed in Theorem 1 and Proposition 2, simulation results of the average AoI is presented since the closed-form expression is not derived, and for the average PAoI, the upper-bound results are compared with the exact simulation results. It is shown that DF relaying scheme outperforms the AF and direct transmissions. For all the three transmission schemes, by increasing $P_t$, both average AoI, and average PAoI decreases because the capacitors charge faster, and the waiting charging time is reduced. For the case of DF relaying with $P_t \rightarrow \infty$, the queue model becomes Geo/Geo/1 and its average PAoI converges to the closed-form presented in Table~\ref{tab:my-table}, and for the AF method according the discussion in Section \ref{sec4S}, the average PAoI also converges to $\frac{1}{P_{suc,AF}}$. 

\begin{figure}[!t]
        \centering
        \psfrag{avarage PAOI} {\!\!\!\!\! Average PAoI}
        \psfrag{pt} {\small \hspace{-5mm}$P_t$}
        \psfrag{DIRECT , drs=10[M]}{\small direct, $d_{d,s}=10$  }
        \psfrag{   AF        , drs=10[M]}{\small AF, $d_{d,s}=10$  }
        \psfrag{   DF        , drs=10[M]}{\small DF, $d_{d,s}=10$  }
        \psfrag{DIRECT , drs=7.5[M]}{\small direct, $d_{d,s}=7.5$   }
        \psfrag{   DF        , drs=7.5[M]}{\small DF, $d_{d,s}=7.5$  }
        \psfrag{   DF        , drs=6.5[M]}{\small DF, $d_{d,s}=6.5$  }
        \psfrag{   AF        , drs=7.5[M]}{\small AF, $d_{d,s}=7.5$  }
        \psfrag{   AF        , drs=6.5[M]}{\small AF, $d_{d,s}=6.5$ }
        \psfrag{DIRECT , drs=6.5[M]}{\small direct, $d_{d,s}=6.5$ }
        \includegraphics[scale=0.64]{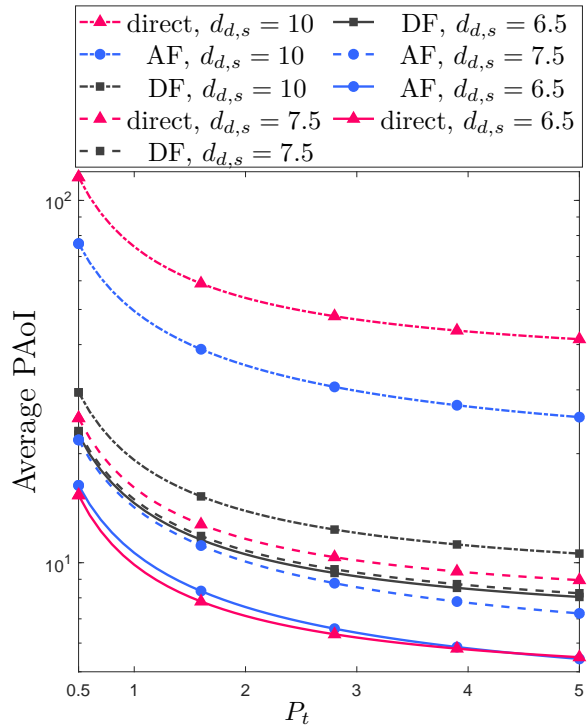}
    
        \caption{Average PAoI versus transmission power $P_t$ for a set of $d_{d,s}$.}
        \label{fig:ptdrs}
\end{figure}

Fig.~\ref{fig:ptdrs} presents the average PAoI as a function of $P_t$ for a set of distance between the source--destination, i.e. $d_{d,s} \in \{10, 7.5, 6.5\}$. As illustrated, by increasing the end-to-end distance of $d_{d,s}$, the average PAoI increases as well. It is shown that, for the case of $d_{d,s}=6.5$, wherein the relay is closer to the destination, the PAoI of the DF scheme is worse than that of the AF one. Interestingly for this case, the average PAoI of the direct transmission only is the best, and thus there is no need to utilize the relay node. However, by increasing $d_{d,s}$ to 7.5, utilizing the relay node becomes more useful, and the cooperative schemes outperform the direct transmission one. Surprisingly, the AF scheme is still achieved less average PAoI than the DF one. Finally, at $d_{d,s}=10$, the DF outperforms the AF cooperative scheme because of more re-transmission is required for the AF strategy due to the larger distance between the source and the destination. 


\begin{figure}[!t]
         \centering
         \psfrag{avarage PAOI} {Average PAoI}
        \psfrag{BT}{\hspace{1mm}$\frac{B_s}{B_r}$}
        \psfrag{DIRECT , gamma=35}{\small direct,$\gamma_{th}=16$ }
        \psfrag{   AF        , gamma=35}{\small AF  , $\gamma_{th}=16$ }
        \psfrag{   DF        , gamma=8}{\small DF  , $\gamma_{th}=10$ }
         \psfrag{   DF        , gamma=18}{\small AF  , $\gamma_{th}=13$ }
        \psfrag{   DF        , gamma=35}{\small DF  , $\gamma_{th}=16$ }
        \psfrag{DIRECT , gamma=18}{\small direct,$\gamma_{th}=13$ }
        \psfrag{DIRECT , gamma=8}{\small AF  , $\gamma_{th}=10$ }
        \psfrag{   AF        , gamma=8} {\small  direct,$\gamma_{th}=10$ }
        \psfrag{   AF        , gamma=18}{\small DF  , $\gamma_{th}=13$}
        \includegraphics[scale=0.61]{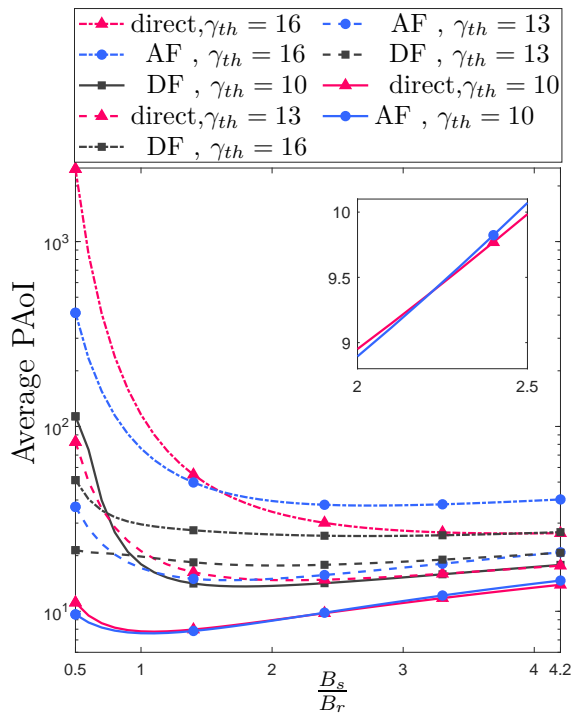}
        \caption{Average PAoI versus capacitors sizes ratio $\frac{B_s}{B_r}$ for different $\gamma_{th}$ in [dB].}
        \label{fig:bsgamma}
\end{figure}
Fig.~\ref{fig:bsgamma} investigates the effect of capacitors sizes on the performance of average PAoI. For the two cooperative DF and AF transmission schemes and the direct one, average PAoI is depicted versus capacitors sizes ratio $\frac{B_s}{B_r}$ for $\gamma_{th} \in \{16,13,10\}$ [dB], $P_t=0.5$, $d_{r,s} = 6$ and $d_{d,s} = 10$. By increasing  $\gamma_{th}$ the average PAoI of AF and direct schemes increase, Higher $\gamma_{th}$ causes lower $p_{suc}$, and thus higher average PAoI is achieved. It is shown that there is an optimal value for the capacitors sizes ratio $\frac{B_s}{B_r}$ to achieve the best average PAoI.By increasing the $\frac{B_s}{B_r}$ ratio, one should wait more, and thus the average PAoI increases. 


Fig.~\ref{fig:pathdi} shows the average PAoI of all the schemes as a function of the path loss exponent $\alpha$ for different values for distances $d_{d,s}$ and $d_{r,s}$. We assume  $P_t = 0.5$, the capacitors sizes ratio $\frac{B_s}{B_r}=1$ and $\gamma_{th}=13$ [dB]. For the cases of AF and direct schemes, by simultaneous increasing of $d_{r,s}$ and $d_{d,s}$ with the same ratio and also increasing the path loss exponent $\alpha$, the average PAoI increases, however for the DF scheme, the average PAoI first decreases and then increases. For $\alpha \geq 2$, direct schemes have the highest average PAoI, so the data in this scheme is most stale. For example, at $d_{d,s}=10$ and $d_{r,s}=6$, for $\alpha \geq 2$, the direct scheme has highest average PAoI. Besides, in this case, the DF scheme has fresher data at a larger amount of the path loss exponent $\alpha$. As can be seen, at higher values of $\alpha$, the average PAoI of the AF and the direct schemes have higher values, so the data would not be fresh anymore. In these situations, DF schemes are always a better choice to transmit real-time information.

\begin{figure}[!t]
        \centering
        \psfrag{avarage PAOI} {\!\!\!\! Average PAoI}
        \psfrag{alpha} {\hspace{2mm}$\alpha$}
        \psfrag{DIRECT , drs=6[m] d=10[m]}{\small direct,$d_{d,s}=10$ $d_{r,s}=6$ }
        \psfrag{DIRECT , drs=4.5[m] d=7.5[m]}{\small direct,$d_{d,s}=7.5$ $d_{r,s}=4.5$}
        \psfrag{   DF       , drs=6[m] d=10[m]}{\small DF, $d_{d,s}=10$ $d_{r,s}=6$}
        \psfrag{   AF        , drs=6[m] d=10[m]}{\small AF,\hspace{0.5mm} $d_{d,s}=10$ $d_{r,s}=6$}
        \psfrag{   AF       , drs=4.5[m] d=7.5[m]}{\small AF, $d_{d,s}=7.5$ $d_{r,s}=4.5$}
        \psfrag{   DF      , drs=4.5[m] d=7.5[m]}{\small DF, $d_{d,s}=7.5$ $d_{r,s}=4.5$ }
        \includegraphics[scale=0.58]{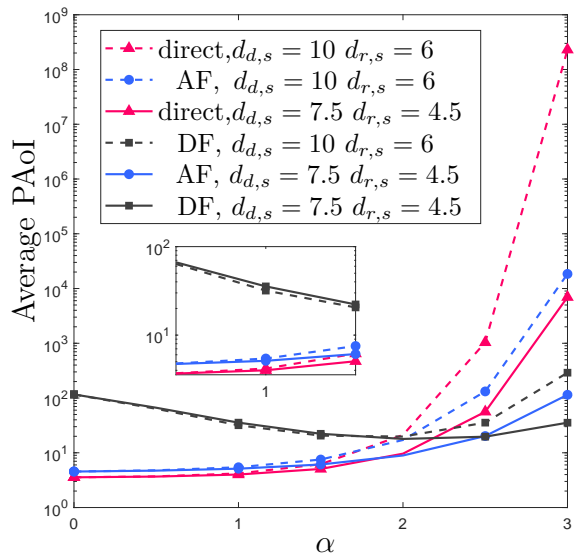}

        \caption{Average PAoI versus path loss exponent $\alpha$.}
        \label{fig:pathdi}     
\end{figure}

In Fig.~\ref{fig:case}, we discuss two main special cases; (a) one-shot successful data transmission, i.e. $P_{suc}=1$ (b) full power nodes, i.e. $P_t \rightarrow \infty$, with $\frac{B_s}{B_r}=2$, $\alpha = 2$,  $d_{d,s} = 10$, and $d_{r,s} =6$. Since performance of the direct scheme is almost the same as that of the AF one, we only investigate AF and DF schemes. For the first case, average PAoI versus $P_t$ is presented in Fig.~\ref{fig:case-a}. The queue model of the DF scheme is P/P/1, and the PAoI is reported in Table~\ref{tab:my-table}. Also for the AF scheme, the average PAoI becomes $\mathbb{E}[T_{AF}]$. It is shown that the analytical results match perfectly with the simulation ones for both AF and DF schemes. So, the upper bound of $\mathbb{E}[W]$ is tight for this case. Moreover, by increasing $P_t \rightarrow \infty$, the queuing model of the DF scheme becomes D/D/1 with average PAoI of 2, and the average PAoI converges to $1$ for the AF scheme as reported in Section IV. Fig.~\ref{fig:case-b} also investigates the average PAoI versus $\gamma_{th}$ for the case of full power nodes. Thus, the queuing type of the DF scheme becomes Geo/Geo/1, and the exact average PAoI is given in (\ref{eg:exact}) and the upper bound is reported in Table~\ref{tab:my-table}. As depicted the upper bound is also tight for this case as well, also for the AF scheme the average PAoI becomes $\frac{1}{P_{suc,AF}}$.

\psfrag{pt}{\footnotesize \hspace{-6mm} $P_t$}
\psfrag{avarage PAOI} {\!\!\!\!\!\!\!\!  Average PAoI}
\psfrag{theory AF}{\small AF}
\psfrag{simulation AF}{\small AF simulation}
\psfrag{simulation DF}{\small DF simulation}
\psfrag{theory DF}{\small DF}
\psfrag{table DF}{\small DF exact}
\begin{figure}[!t]
    \centering
    \subfloat[one-shot successful data transmission]{\label{fig:case-a}{\includegraphics[scale=0.5]{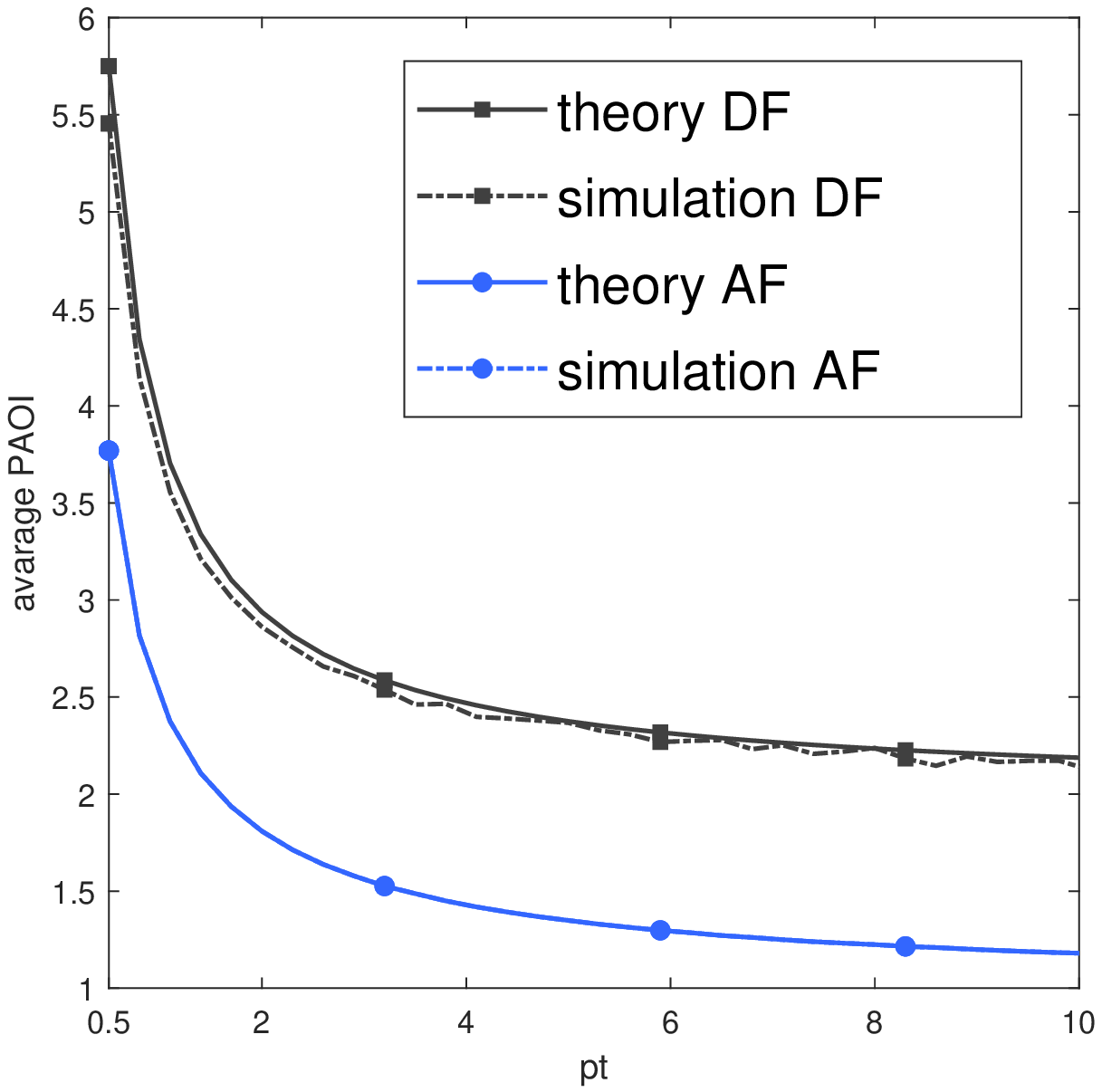} }}%
    \psfrag{gamma} {\hspace{2mm}\small $\gamma_{th}$ [dB]}
     \psfrag{avarage PAOI} {\!\!\!\!\!\!\!\! Average PAoI}
    \qquad
 \subfloat[full-power  nodes]{\label{fig:case-b}{\includegraphics[scale=0.5]{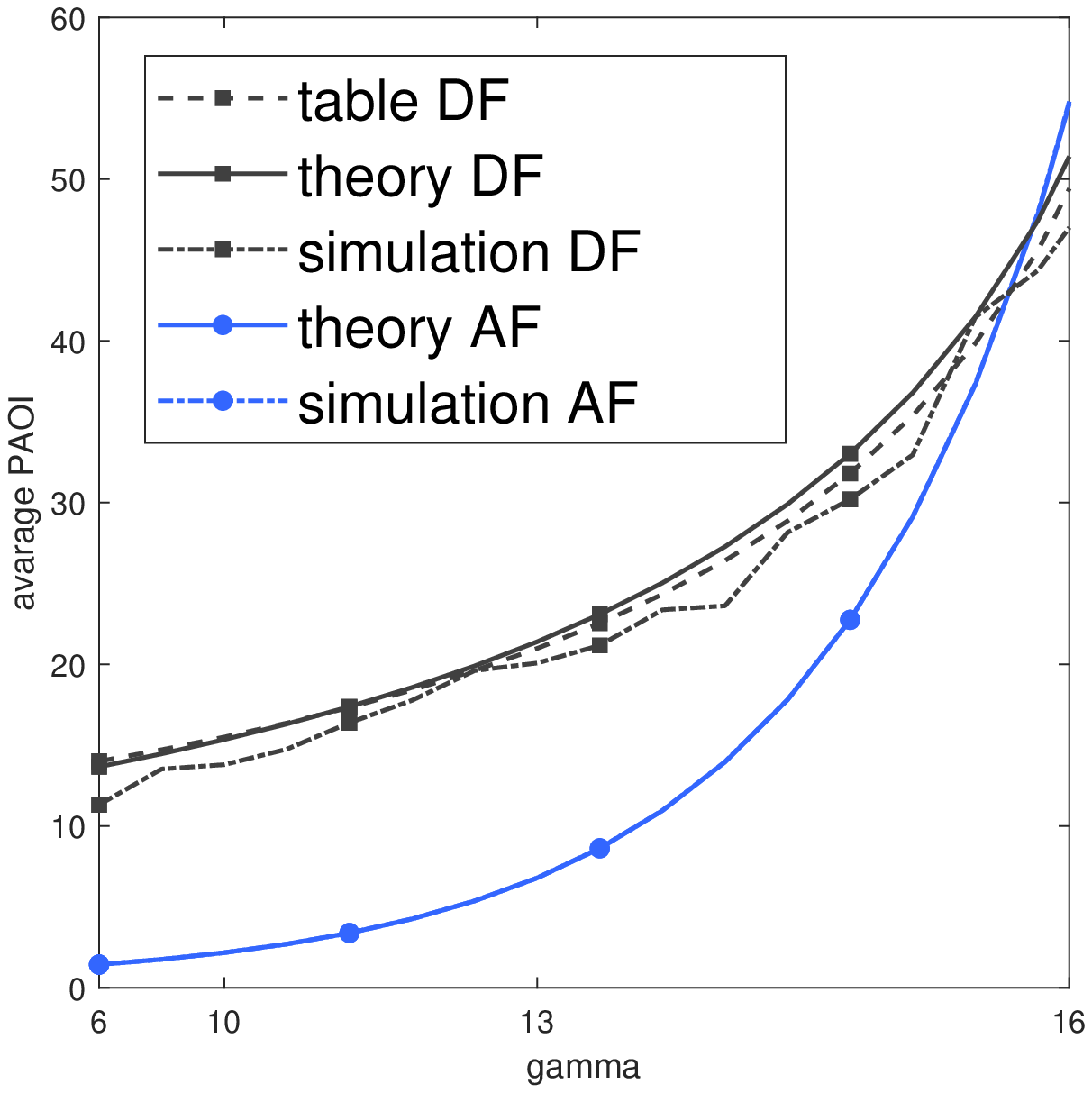} }}
         
         \psfrag{simulation AF pai}{\small AF simulation}
        \psfrag{simulation DF pai}{\small DF simulation}
        \psfrag{theory DF}{\small DF}
        \psfrag{table DF}{\small DF excat}
    \caption{Average PAoI for two special cases.}
    \label{fig:case}
\end{figure}

\section{Conclusion}\label{sec6con}
To address whether a cooperative relaying is still beneficial from data freshness perspectives, we studied average AoI and PAoI metrics for a wireless powered cooperative communications system, wherein the source and the relay had finite-sized capacitors to harvest energy from the wireless power signal transmitted from a power station. Since the wireless power and information transfer channels were modeled by flat fading ones, two main randomnesses are affecting the age characteristics; the first one is the random waiting time to fully charge the capacitors of the source and the relay, and the second one is random numbers of data re-transmission to deliver a packet to the destination successfully. It was shown that for the DF relaying scheme, the distribution of the waiting time of the system had a general form, and thus a closed-form upper bound of the average PAoI was presented, and one could compute average AoI numerically, as well.  However, the average AoI and PAoI were analytically derived for the AF scheme. In addition, all special cases of the proposed system model are discussed, and the simplified results were presented to highlight the extreme performances of the general system model. We concluded that utilizing a relay not only enhances the reliability of the end-to-end communications but also could improve the freshness of data at the destination. Finally, numerical results were presented to clarify which of the DF or AF schemes is more beneficial compared to the direct transmission, i.e. without utilizing a cooperative relay. Moreover, the average PAoI were analysed as a function of transmission power, relay position, minimum SNR to have a successful reception, capacitors sizes ratio, and path loss exponent.

\section*{Appendix A\\Proof of Theorem \ref{thm:FPA}}
\begin{proof} As discussed before, for the proposed system model, we have a G/G/1 queue (also known as GI/G/1 in the literature) with mean arrival rate $\lambda:=1/\mathbb{E}[X_s]$ and mean service rate $\mu:=1/\mathbb{E}[X_r]$. Therefore, to ensure the stability of the queue, the utility $\rho:=\frac{\lambda}{\mu}$ must be less than $1$, that is, we must have $ \mathbb{E}[X_r] < \mathbb{E}[X_s]$ to have a stable queue.

As the average PAoI is presented in (\ref{eq:df}), we firstly discuss evaluating $\mathbb{E}[W]$. Since the queue is G/G/1, it is not straight forward to compute the exact value of $\mathbb{E}[W]$. Thus, we use the following upper bound for evaluating the mean waiting time \cite{kingman, marshall}
\begin{equation}
    \mathbb{E}[W] \le \frac{\lambda (\sigma^2_{X_s} + \sigma^2_{X_r})}{2(1-\rho)},
\end{equation}
where $\lambda$ is mean arrival rate and $\sigma^2_{X_s}$, $\sigma^2_{X_r}$ are variances of interarrival time and service time, respectively. So, we have
\begin{equation}
    \mathbb{E}[W] \le \frac{\Big (\mathbb{E}[X^2_s]-\mathbb{E}^2[X_s] \Big) + \Big (\mathbb{E}[X^2_r]-\mathbb{E}^2[X_r] \Big)}{2 \Big( \mathbb{E}[X_s]-\mathbb{E}[X_r]\Big)}.\label{eq:upper bound}
\end{equation} 

Therefore, to evaluate the average PAoI $A_{DF} = \mathbb{E}[X_s +X_r +W]$, we just need to compute the mean and variance of $X_l$ for $l \in \{s,r\}$. Following the same steps as \cite{kirikidis}, the first-order and second-order moments of the time between the successful reception are given by
\begin{subequations}
\begin{align}
\mathbb{E}[X_l] &= \sum_{n=1}^{\infty}n\mathbb{E}[T_l]{(1-P_{suc,l})}^{(n-1)}P_{suc,l}\nonumber \\ &=\frac{\mathbb{E}[T_l]}{P_{suc,l}}, \label{eq:subeq1}
\\ \mathbb{E}[X_l^2] &= \mathbb{E} \left[ (\sum_{n=1}^{\infty}T^n_l)^2   \right] \nonumber \\ &= \frac{\mathbb{E}[T_l^2]}{P_{suc,l}}+\frac{2(\mathbb{E}[T_l])^2(1-P_{suc,l})}{P_{suc,l}^2},\label{eq:subeq2}
\end{align}
\end{subequations}
where $\mathbb{E}[T_l]$ and $\mathbb{E}[T_l^2]$ are derived in Preposition~\ref{thm:FPA}.  Thus, by plugging  (\ref{eq:subeq1}), (\ref{eq:subeq2}) in (\ref{eq:upper bound}) and (\ref{eq:df}), Theorem is proved. 
\end{proof}

\section*{Appendix B\\Proof of Proposition \ref{props:3}}
For the AF relaying scenario, after some standard mathematical manipulations, the probability mass function of the total waiting time $T_{AF}=\max(T_s,T_r)$ becomes
\begin{align}
    P(T_{AF}=m) =&\exp{\Big(-(B_s^\prime B_r^\prime)\Big)}\left(\frac{(B_s^\prime B_r^\prime)^{m-1}}{(m-1)!^2}+\frac{(B_r^\prime)^{m-1}}{(m-1)!}\sum_{i=1}^{m-1}\frac{(B_s^\prime)^{i-1}}{(i-1)!}+\frac{(B_s^\prime)^{m-1}}{(m-1)!}\sum_{i=1}^{m-1}\frac{(B_r^\prime)^{i-1}}{(i-1)!}\right)\nonumber
    \\\nonumber
    =&\exp{\Big(-(B_s^\prime B_r^\prime)\Big)}\frac{(B_s^\prime  B_r^\prime)^{m-1}}{(m-1)!^2}+\exp{(-B_s^\prime)}\frac{(B_s^\prime)^{m-1}\Gamma(m-1,B_r^\prime)}{(m-1)!(m-2)!}
    \\ & +\exp{(-B_r^\prime)}\frac{(B_r^\prime)^{m-1}\Gamma(m-1,B_s^\prime)}{(m-1)!(m-2)!}.
\end{align}
On the other hand, we know that for any non-negative discrete real random variable $T$, instead of using $\mathbb{E}[T] = \sum_{i=0}^{\infty} iP(T=i)$ and $\mathbb{E}[T^2] = \sum_{i=0}^{\infty} i^2P(T=i)$, we can use the following equivalent relations;
\begin{subequations}
\begin{equation}
    \mathbb{E}[T]=\sum_{i=0}^{\infty}\Big(1-F_T(i)\Big), \label{eq:et}
\end{equation}
\begin{equation}
  \mathbb{E}[T^2]=2\sum_{i=0}^{\infty}i\Big(1-F_T(i)\Big),\label{eq:et2}
\end{equation}
\end{subequations}
where $F_T(.)$ is the cumulative distribution function (CDF) of the random variable. The CDF of $T_{AF}$ is
\begin{align}\label{eq:cdf}
   F_{T_{AF}}(k) &=\mathbb{P}_r \left \{T_{AF} \le k \right \} \nonumber 
   \\ &= \mathbb{P}_r\left \{\max(T_s,T_r) \le k\right \} \nonumber \\ &= \mathbb{P}_r\left \{T_s\le k,T_r \le k\right \} \nonumber
   \\ &= \mathbb{P}_r\left \{T_s \le k\right \} \mathbb{P}_r\left \{T_r \le k\right \} \nonumber 
   \\ &=
   \Big( \exp(-B_s^{\prime})\sum_{m=1}^{k}\frac{{(B_s^{\prime})}^{m-1}}{(m-1)!} \Big) \Big( \exp(-B_r^{\prime})\sum_{m=1}^{k}\frac{{(B_r^{\prime})}^{m-1}}{(m-1)!}\Big). 
\end{align}
We know that  $ ~\frac{\Gamma(n,x)}{\Gamma(n)}=e^{-x}\sum_{d=0}^{n-1}\frac{x^d}{d!}~$, so we can rewrite (\ref{eq:cdf}) as
\textcolor{black}{
\begin{equation}\label{eq:lastcdf}
    F_{T_{AF}}(k)=\frac{\Gamma(k,B_s^\prime)\Gamma(k,B_r^\prime)}{\Gamma(k)^2}.
\end{equation}}
By substituting (\ref{eq:lastcdf}) in (\ref{eq:et}) and (\ref{eq:et2}) we could derive (\ref{eq:etlast}) and (\ref{eq:et2last}), respectively.

\bibliographystyle{IEEEtran}
\bibliography{refrences}

\end{document}